\documentclass{llncs}
\usepackage{tabls}
\usepackage{graphicx}
\usepackage{wrapfig}
\usepackage{array}
\usepackage[algoruled,linesnumbered,algo2e,vlined]{algorithm2e}
\usepackage{url}
\usepackage{amssymb}
\usepackage{multirow}
\usepackage{multicol}
\usepackage{color} 

\newcommand{\Occ}{\mathit{Occ}}
\newcommand{\makeSuffixArray}{\mathit{SUFFIXARRAY}}
\newcommand{\SA}{\mathit{SA}}
\newcommand{\makeLCPArray}{\mathit{LCPARRAY}}
\newcommand{\LCP}{\mathit{LCP}}
\newcommand{\derive}{\mathit{val}}
\newcommand{\deriveInt}{\mathit{itv}}
\newcommand{\VarOcc}{\mathit{vOcc}}
\newcommand{\suffix}{\mathit{suf}}
\newcommand{\prefix}{\mathit{pre}}
\newlength\savedwidth
\newcommand{\wcline}[1]{\noalign{\global\savedwidth\arrayrulewidth\global\arrayrulewidth 1.5pt} \cline{#1}
\noalign{\global\arrayrulewidth\savedwidth}}

\newcommand{\ignore}[1]{}

\author{
  Keisuke Goto 
  \and
  Hideo Bannai 
  \and
  Shunsuke Inenaga 
  \and
  Masayuki Takeda 
}
\institute{
  Department of Informatics, Kyushu University
  \email{\{keisuke.gotou,bannai,inenaga,takeda\}@inf.kyushu-u.ac.jp}\\
}
\title{
  Fast $q$-gram Mining on SLP~Compressed~Strings\thanks{This work was
    supported by KAKENHI 22680014 (HB)}
}
\date{}

\begin{document}
\maketitle
\begin{abstract}
  We present simple and efficient algorithms for calculating
  $q$-gram frequencies on strings represented in compressed form,
  namely, as a straight line program (SLP).
  Given an SLP of size $n$ that represents string $T$,
  we present an $O(qn)$ time and space algorithm
  that computes the occurrence frequencies of
  {\em all} $q$-grams in $T$.
  Computational experiments show that our algorithm and its variation
  are practical for small $q$, actually running faster on various
  real string data, compared to algorithms that work on the
  uncompressed text.
  We also discuss applications in data mining and classification of
  string data, for which our algorithms can be useful.
\end{abstract}
\section{Introduction}
A major problem in managing large scale string data is its sheer size.
Therefore, such data is normally stored in compressed form.
In order to utilize or analyze the data afterwards, 
the string is usually decompressed, where we must again confront the
size of the data.
To cope with this problem, algorithms that work directly on compressed
representations of strings without explicit decompression
have gained attention, especially for the string pattern matching problem~\cite{amir92:_effic_two_dimen_compr_match}
where
algorithms on compressed text can actually run faster than
algorithms on the uncompressed text~\cite{ShibataCIAC2000}.
There has been growing interest in what problems can be
efficiently solved in this kind of 
setting~\cite{lifshits07:_proces_compr_texts,hermelin09:_unified_algor_accel_edit_distan}.

Since there exist many different text compression schemes,
it is not realistic to develop different algorithms for each scheme.
Thus, it is common to consider algorithms on texts represented as
{\em straight line programs} (SLPs)~\cite{NJC97,lifshits07:_proces_compr_texts,hermelin09:_unified_algor_accel_edit_distan}.
An SLP is a context free grammar in the Chomsky normal form that
derives a single string.
Texts compressed by any grammar-based compression algorithms 
(e.g.~\cite{SEQUITUR,LarssonDCC99}) can be represented as SLPs,
and those compressed by the LZ-family (e.g.~\cite{LZ77,LZ78}) can be quickly
transformed to SLPs~\cite{rytter03:_applic_lempel_ziv}.
Recently, even {\em compressed self-indices} based on SLPs have
appeared~\cite{claudear:_self_index_gramm_based_compr},
and SLPs are a promising representation of compressed strings
for conducting various operations.

In this paper, we explore a more advanced field of
application for compressed string processing:
mining and classification on
string data given in compressed form.
Discovering useful patterns hidden in strings as well as
automatic and accurate classification of
strings into various groups, are important problems in
the field of data mining and machine learning with many
applications.
As a first step toward {\em compressed} string mining and
classification, we consider the problem of
finding the occurrence frequencies for all
$q$-grams contained in a given string. 
$q$-grams are important features of string data,
widely used for this purpose in many fields such as text and natural
language processing, and bioinformatics.

In~\cite{inenaga09:_findin_charac_subst_compr_texts},
an $O(|\Sigma|^2n^2)$-time $O(n^2)$-space algorithm
for finding the {\em most frequent} $2$-gram from an 
SLP of size $n$ representing text $T$ over alphabet $\Sigma$ was presented.
In~\cite{claudear:_self_index_gramm_based_compr},
it is mentioned
that the most frequent $2$-gram can be found in $O(|\Sigma|^2n\log n)$-time
and $O(n\log|T|)$-space, if the SLP is pre-processed and a self-index is built. 
It is possible to extend these two algorithms
to handle $q$-grams for $q > 2$, but would respectively require
$O(|\Sigma|^qqn^2)$ and $O(|\Sigma|^qqn\log n)$ time,
since they must essentially enumerate and count the occurrences of all
substrings of length $q$,
regardless of whether the $q$-gram occurs in the string.
Note also that any algorithm that works on the uncompressed text
$T$ requires exponential time in the worst case, since $|T|$ can be as
large as $O(2^n)$.

The main contribution of this paper is an $O(qn)$ time and space
algorithm that computes the occurrence frequencies for {\em all}
$q$-grams in the text, given
an SLP of size $n$ representing the text.
Our new algorithm solves the more general problem
and greatly improves the computational complexity compared to previous work.
We also conduct computational experiments on various real texts,
showing that when $q$ is small, our algorithm and its variation actually run faster
than algorithms that work on the uncompressed text.

Our algorithms have profound applications in the field of string mining
and classification, and several applications and extensions are discussed.
For example, our algorithm leads to an $O(q(n_1+n_2))$ time algorithm for
computing the $q$-gram spectrum kernel~\cite{leslie02:_spect_kernel}
between SLP compressed texts of size $n_1$ and $n_2$.
It also leads to an $O(qn)$ time algorithm for finding
the optimal $q$-gram (or emerging $q$-gram) that discriminates between
two sets of SLP compressed strings, when $n$ is the total size of the
SLPs.


\subsubsection{Related Work}
There exist many works on {\em compressed text
indices}~\cite{navarro07:_compr},
but the main focus there is on fast search for a {\em given} pattern.
The compressed indices basically replace or simulate operations on
uncompressed indices using a smaller data structure.
Indices are important for efficient string processing, 
but note that simply replacing the underlying index used in a mining
algorithm will generally increase time complexities of the algorithm
due to the extra overhead required to access the compressed index.
On the other hand, our approach is a new mining algorithm which exploits
characteristics of the compressed representation to achieve faster
running times.

Several algorithms for finding characteristic sequences from compressed texts have been proposed, e.g.,
finding the longest common substring of two strings~\cite{matsubara_tcs2009},
finding all palindromes~\cite{matsubara_tcs2009}, 
finding most frequent substrings~\cite{inenaga09:_findin_charac_subst_compr_texts}, and
finding the longest repeating substring~\cite{inenaga09:_findin_charac_subst_compr_texts}.
However, none of them have reported results of computational
experiments, implying that this paper is the first to show the
practical usefulness of a compressed text mining algorithm.

\section{Preliminaries}
Let $\Sigma$ be a finite {\em alphabet}.
An element of $\Sigma^*$ is called a {\em string}.
For any integer $q > 0$, an element of $\Sigma^q$ is called an \emph{$q$-gram}.
The length of a string $T$ is denoted by $|T|$. 
The empty string $\varepsilon$ is a string of length 0,
namely, $|\varepsilon| = 0$.
For a string $T = XYZ$, $X$, $Y$ and $Z$ are called
a \emph{prefix}, \emph{substring}, and \emph{suffix} of $T$, respectively.
The $i$-th character of a string $T$ is denoted by $T[i]$ for $1 \leq i \leq |T|$,
and the substring of a string $T$ that begins at position $i$ and
ends at position $j$ is denoted by $T[i:j]$ for $1 \leq i \leq j \leq |T|$.
For convenience, let $T[i:j] = \varepsilon$ if $j < i$.

For a string $T$ and integer $q \geq 0$, let $\prefix(T,q)$ and $\suffix(T,q)$
represent respectively, the length-$q$ prefix and suffix of $T$.
That is, $\prefix(T,q) = T[1:\min(q,|T|)]$ and
$\suffix(T,q) = T[\max(1,|T|-q+1):|T|]$.

For any strings $T$ and $P$,
let $\Occ(T,P)$ be the set of occurrences of $P$ in $T$, i.e.,
$\Occ(T,P) = \{k > 0 \mid T[k:k+|P|-1] = P\}$.
The number of elements $|\Occ(T,P)|$ is called
the \emph{occurrence frequency} of $P$ in $T$.

\subsection{Straight Line Programs}
\begin{algorithm2e}[t]
  \caption{Calculating $\VarOcc(X_i)$ for all $1\leq i \leq n$.}
  \label{algo:varocc}
  \KwIn{SLP ${\mathcal T} = \{X_i\}_{i=1}^n$ representing string $T$.}
  \KwOut{$\VarOcc(X_i)$  for all $1\leq i\leq n$}
  $\VarOcc[X_n] \leftarrow 1$\;
  \lFor{$i\leftarrow 1$ \KwTo $n-1$}{
    $\VarOcc[X_i] \leftarrow 0$\;
  }
  \For{$i\leftarrow n$ \KwTo $2$}{
    \If{ $X_i = X_{\ell}X_r$ }{
      $\VarOcc[X_{\ell}] \leftarrow \VarOcc[X_{\ell}] + \VarOcc[X_i]$;
      $\VarOcc[X_r] \leftarrow \VarOcc[X_r] + \VarOcc[X_i]$\;
    }
  }
\end{algorithm2e}

\begin{wrapfigure}[13]{r}{0.5\textwidth}
  \vspace{-1.56cm}
\centerline{\includegraphics[width=0.48\textwidth]{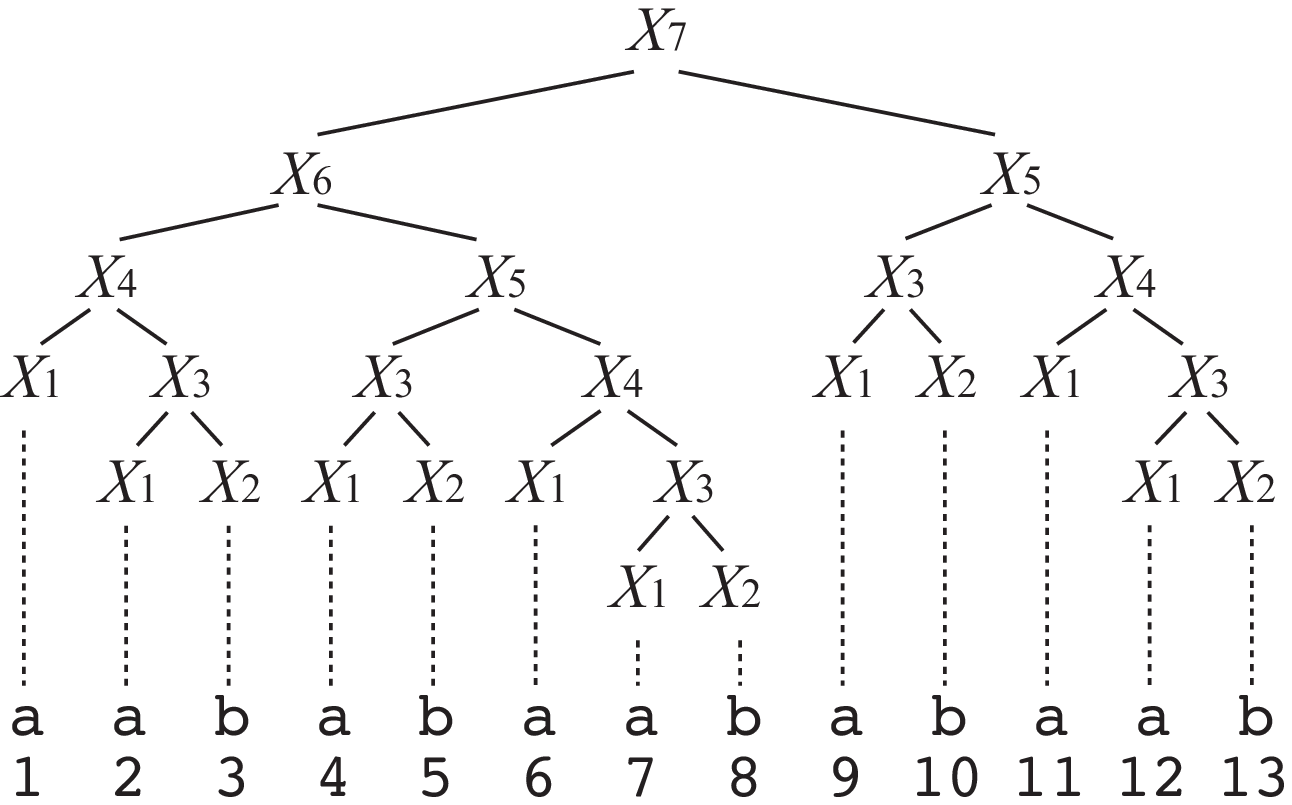}}
\caption{
  The derivation tree of
  SLP $\mathcal T = \{X_i\}_{i=1}^{7}$ with $X_1 = \mathtt{a}$, $X_2 = \mathtt{b}$, $X_3 = X_1X_2$,
  $X_4 = X_1X_3$, $X_5 = X_3X_4$, $X_6 = X_4X_5$, and $X_7 = X_6X_5$,
  representing string $T = \derive(X_7) = \mathtt{aababaababaab}$.
}
\label{fig:SLP}
\end{wrapfigure}

A {\em straight line program} ({\em SLP}) $\mathcal T$ is a sequence of assignments 
$X_1 = expr_1, X_2 = expr_2, \ldots, X_n = expr_n$,
where each $X_i$ is a variable and each $expr_i$ is an expression, where
$expr_i = a$ ($a\in\Sigma$), or $expr_i = X_{\ell} X_r$~($\ell,r < i $).
Let $\derive(X_i)$ represent the string derived from $X_i$.
When it is not confusing, we identify a variable $X_i$
with $\derive(X_i)$.
Then, $|X_i|$ denotes the length of the string $X_i$ derives.
An SLP $\mathcal{T}$ {\em represents} 
the string $T = \derive(X_n)$.
The \emph{size} of the program $\mathcal T$ is the number $n$ of
assignments in $\mathcal T$. (See Fig.~\ref{fig:SLP})


The substring intervals of $T$ that each variable derives can 
be defined recursively as follows: $\deriveInt(X_n) = \{ [1:|T|]\}$, and
$\deriveInt(X_i) =
\{ [u + |X_\ell|:v] \mid X_k = X_\ell X_i, [u:v] \in \deriveInt(X_k) \}
\cup
\{ [u:u+|X_i|-1] \mid X_k = X_i X_r, [u:v] \in \deriveInt(X_k) \}
$ for $i < n$.
For example,  $\deriveInt(X_5) = \{ [4:8], [9:13]\}$ in Fig.~\ref{fig:SLP}.
Considering the transitive reduction of set inclusion,
the intervals $\cup_{i=1}^n \deriveInt(X_i)$ naturally form a
binary tree (the derivation tree).
Let $\VarOcc(X_i) = |\deriveInt(X_i)|$ denote 
the number of times a variable $X_i$ occurs in the derivation of $T$.
$\VarOcc(X_i)$ for all $1\leq i\leq n$ can be computed in $O(n)$
time by a simple iteration on the variables,
since $\VarOcc(X_n) = 1$ and for $i < n$,
$\VarOcc(X_i) =
\sum \{ \VarOcc(X_k) \mid X_k = X_\ell X_i \}
+
\sum \{\VarOcc(X_k) \mid X_k = X_iX_r \}$.
(See Algorithm~\ref{algo:varocc})

\subsection{Suffix Arrays and LCP Arrays}
The suffix array $\SA$~\cite{manber93:_suffix} of any string $T$
is an array of length $|T|$ such that
$\SA[i] = j$, where $T[j:|T|]$ is the $i$-th lexicographically smallest suffix of $T$.
The \emph{lcp} array of any string $T$ is an array of length $|T|$ such that
$\LCP[i]$ is the length of the longest common prefix of
$T[\SA[i-1]:|T|]$ and $T[\SA[i]:|T|]$ for $2 \leq i \leq |T|$, 
and $\LCP[1] = 0$.
The suffix array for any string of length $|T|$
can be constructed in $O(|T|)$ 
time~(e.g.~\cite{Karkkainen_Sanders_icalp03})
assuming an integer alphabet.
Given the text and suffix array, the lcp array can also be calculated
in $O(|T|)$ time~\cite{Kasai01}.




\section{Algorithm}
\subsection{Computing $q$-gram Frequencies on Uncompressed Strings}
\label{subsection:uncompressed}
We describe two algorithms (Algorithm~\ref{algo:naive}
and Algorithm~\ref{algo:naive_sa}) for computing
the $q$-gram frequencies of a given uncompressed string $T$.

A na\"ive algorithm for computing the
$q$-gram frequencies is given in Algorithm~\ref{algo:naive}.
The algorithm constructs an associative array, where keys consist
of $q$-grams, and the values correspond to the occurrence frequencies
of the $q$-grams.
The time complexity depends on the implementation of the associative
array,
but requires at least $O(q|T|)$ time 
since each $q$-gram is considered explicitly, and
the associative array is accessed $O(|T|)$ times:
e.g. $O(q|T|\log|\Sigma|)$ time and $O(q|T|)$ space using a
simple trie.
\begin{algorithm2e}[t]
      \SetKwInput{KwOut}{Report}
      \SetKw{KwReport}{Report}
      \KwIn{string $T$, integer $q \geq 1$}
      \KwOut{$(P, |\Occ(T,P)|)$ for all $P\in\Sigma^q$ where
        $\Occ(T,P) \neq \emptyset$.}
      $\mathbf{S} \leftarrow \emptyset$\tcp*[l]{empty associative array}
      \For{$i\leftarrow 1$ \KwTo $|T|-q+1$}{
        $\mathit{qgram} \leftarrow T[i:i+q-1]$\;
        \lIf{$\mathit{qgram} \in\mathrm{keys}(\mathbf{S})$}{\label{algo:naive:increment}
          $\mathbf{S}[\mathit{qgram}] \leftarrow
          \mathbf{S}[\mathit{qgram}] + 1$\;      
        }
        \lElse{
          $\mathbf{S}[\mathit{qgram}] \leftarrow 1$\tcp*[l]{new $q$-gram}\label{algo:naive:addweight}
        }
      }
      \lFor{$\mathit{qgram}\in\mathrm{keys}(\mathbf{S})$}{\KwReport{$(\mathit{qgram},\mathbf{S}[\mathit{qgram}])$}}
      \caption{A na\"ive algorithm for computing $q$-gram frequencies.}
      \label{algo:naive}
      \end{algorithm2e}
    \begin{algorithm2e}[t]
      \SetKwInput{KwOut}{Report}
      \SetKw{KwReport}{Report}
      \SetKw{KwOr}{or}
      \SetKw{KwAnd}{and}
      \KwIn{string $T$, integer $q \geq 1$}
      \KwOut{
        $(i, |\Occ(T,P)|)$
        for all $P\in\Sigma^q$ and some position $i\in \Occ(T,P)$.
      }
      $\SA \leftarrow \makeSuffixArray(T)$;
      $\LCP \leftarrow \makeLCPArray(T,SA)$;
      $\mathit{count} \leftarrow 1$\;  
      \For{$i\leftarrow 2$ \KwTo $|T|+1$}{
        \If{$i = |T|+1$ \KwOr $\LCP[i] < q$}{
          \lIf{$\mathit{count} > 0$}{
            \KwReport{$(\SA[i-1], \mathit{count})$}; $\mathit{count} \leftarrow 0$\;
          }          
        }
        \lIf{$i \leq |T|$ \KwAnd $\SA[i] \leq |T| - q + 1$}{
          $\mathit{count} \leftarrow \mathit{count}+1$\;\label{algo:naive_sa:increment}
        }
      }
      \caption{A linear time algorithm for computing $q$-gram frequencies.}
      \label{algo:naive_sa}
    \end{algorithm2e}

The $q$-gram frequencies of string $T$ can be calculated in
$O(|T|)$ time using suffix array $\SA$ and lcp array $\LCP$,
as shown in Algorithm~\ref{algo:naive_sa}.
For each $1\leq i\leq |T|$, the suffix $\SA[i]$ represents
an occurrence of $q$-gram $T[\SA[i]:\SA[i]+q-1]$, 
if the suffix is long enough, i.e. $\SA[i] \leq |T| - q +1$.
The key is that since the suffixes are lexicographically
sorted, intervals on the suffix array where the values in the lcp array
are at least $q$ represent occurrences of the same $q$-gram.
The algorithm runs in $O(|T|)$ time, since $\SA$ and $\LCP$ can be
constructed in $O(|T|)$. The rest is a simple $O(|T|)$ loop.
A technicality is that we encode the output for a $q$-gram as
one of the positions in the text where the $q$-gram occurs,
rather than the $q$-gram itself.
This is because there can be a total of $O(|T|)$ different $q$-grams,
and if we output them as length-$q$ strings,
it would require at least $O(q|T|)$ time.

\subsection{Computing $q$-gram Frequencies on SLP}

We now describe the core idea of our algorithms,
and explain two variations which utilize variants of the two algorithms for uncompressed strings 
presented in Section~\ref{subsection:uncompressed}.
For $q = 1$, the $1$-gram frequencies are simply the frequencies of
the alphabet and the output is
$(a,\sum\{\VarOcc(X_i)\mid X_i = a\})$ for each $a\in\Sigma$,
which takes only $O(n)$ time.
For $q \geq 2$, we make use of Lemma~\ref{lemma:mkSLP} below.
The idea is similar to the {\em $mk$ Lemma}~\cite{charikar05:_small_gramm_probl},
but the statement is more specific.
\begin{lemma}
  \label{lemma:mkSLP}
  Let $\mathcal{T} = \{X_i\}_{i=1}^n$ be an SLP that represents string $T$.
  For an interval $[u:v]$ $(1\leq u < v \leq |T|)$,
  there exists exactly one variable $X_i=X_\ell X_r$
  such that for some $[u':v'] \in \deriveInt(X_i)$,
  the following holds:
  $[u:v] \subseteq [u':v']$,
  $u \in [u':u'+|X_\ell|-1] \in \deriveInt(X_\ell)$
  and 
  $v \in [u'+|X_\ell|:v'] \in \deriveInt(X_r)$.
\end{lemma}
\begin{proof}
  Consider length $1$ intervals $[u:u]$ and $[v:v]$
  corresponding to leaves in the derivation tree.
  $X_i$ corresponds to the lowest common ancestor
  of these intervals in the derivation tree.
  \qed
\end{proof}

\begin{wrapfigure}[9]{r}{0.36\textwidth}
  \vspace{-1.43cm}
  \centerline{\includegraphics[width=0.35\textwidth]{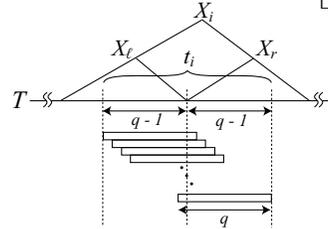}}
  \caption{
    Length-$q$ intervals corresponding to $X_i = X_\ell X_r$.
  }
  \label{fig:SLP-kgram}
\end{wrapfigure}

From Lemma~\ref{lemma:mkSLP}, each occurrence of
a $q$-gram ($q \geq 2$) represented by some length-$q$ interval of $T$,
corresponds to a single variable $X_i = X_{\ell}X_r$, and is split in
two by intervals corresponding to $X_{\ell}$ and $X_r$.
On the other hand, consider all length-$q$ intervals that
correspond to a given variable. Counting the frequencies of the $q$-grams
they represent, and summing them up for all variables give the
frequencies of all $q$-grams of $T$.

%
For variable $X_i = X_\ell X_r$, let 
$t_i = \suffix(X_\ell, q-1) \prefix(X_r,q-1)$.
Then, all $q$-grams represented by length $q$ intervals that 
correspond to $X_i$ are those in $t_i$. (Fig.~\ref{fig:SLP-kgram}).
If we obtain the frequencies of all $q$-grams in $t_i$, and then
multiply each frequency by $\VarOcc(X_i)$, 
we obtain frequencies for the $q$-grams occurring in all intervals
derived by $X_i$.
It remains to sum up the $q$-gram frequencies of $t_i$
for all $1\leq i\leq n$.
We can regard it as obtaining the {\em weighted} $q$-gram
frequencies in the set of strings $\{t_1,\ldots,t_n\}$,
where each $q$-gram in $t_i$ is weighted by $\VarOcc(X_i)$.

We further reduce this problem to a weighted $q$-gram frequency
problem for a single string $z$ as in Algorithm~\ref{algo:slpmain}.
String $z$ is constructed by concatenating $t_i$ 
such that $q \leq |t_i| \leq 2(q-1)$,
and the weights of $q$-grams starting at each position in $z$
is held in array $w$.
On line~\ref{algo:slpmain:append0}, $0$'s instead of $\VarOcc(X_i)$ are appended to
$w$ for the last $q-1$ values corresponding to $t_i$.
This is to avoid counting unwanted $q$-grams that are generated by the
concatenation of $t_i$ to $z$ on line~\ref{algo:slpmain:zappend},
which are not substrings of each $t_i$.
The weighted $q$-gram frequency problem for a single string
(Line~\ref{algo:slpmainweightedfreqs}) can be solved
with a slight modification of Algorithm~\ref{algo:naive} or \ref{algo:naive_sa}.
The modified algorithms are shown respectively in
Algorithms~\ref{algo:weighted_naive} and \ref{algo:weighted_naive_sa}.

\begin{theorem}
  Given an SLP ${\mathcal T} = \{X_i\}_{i=1}^n$ of size $n$
  representing a string $T$,
  the $q$-gram frequencies of $T$ can be computed in $O(qn)$ time
  for any $q > 0$.
\end{theorem}
\begin{proof}
  Consider Algorithm~\ref{algo:slpmain}.
  The correctness is straightforward from the above arguments, so we
  consider the time complexity.
  Line~\ref{algo:slpmain:varocc} can be computed in $O(n)$ time.
  Line~\ref{algo:slpmain:prefsuf} can be computed in $O(qn)$ time by a simple dynamic
  programming.
  For $\prefix()$:
  If $X_i=a$ for some $a\in\Sigma$, then $\prefix(X_i,q-1) = a$.
  If $X_i=X_{\ell}X_r$ and $|X_{\ell}| \geq q-1$, then $\prefix(X_i, q-1) = \prefix(X_{\ell}, q-1)$.
  If $X_i=X_{\ell}X_r$ and $|X_{\ell}| < q-1$, then    $\prefix(X_i, q-1)
  = \prefix(X_{\ell}, q-1)\prefix(X_r, q - 1 - |X_{\ell}|)$.
  The strings $\suffix()$ can be computed similarly.
  The computation amounts to copying $O(q)$ characters for each
  variable, and thus can be done in $O(qn)$ time.
%
  For the loop at line~\ref{algo:slpmain:mainloop}, since the length of
  string $t_i$ appended to $z$, as well as the number of elements
  appended to $w$ is at most $2(q-1)$ in each loop,
  the total time complexity is $O(qn)$.
  Finally, since the length of $z$ and $w$ is $O(qn)$,
  line~\ref{algo:slpmainweightedfreqs} can be calculated in $O(qn)$
  time
  using the weighted version of Algorithm~\ref{algo:naive_sa} (Algorithm~\ref{algo:weighted_naive_sa}).
  \qed
\end{proof}
Note that the time complexity for using the weighted version of
Algorithm~\ref{algo:naive} for line~\ref{algo:slpmainweightedfreqs}
of Algorithm~\ref{algo:slpmain} would be at least $O(q^2n)$:
e.g. $O(q^2n\log|\Sigma|)$ time and $O(q^2n)$ space using a trie.

\begin{algorithm2e}[t]
  \caption{Calculating $q$-gram frequencies of an SLP for $q\geq 2$}
  \label{algo:slpmain}
  \SetKw{KwAnd}{and}
  \SetKwInput{KwOut}{Report}
  \KwIn{SLP ${\mathcal T} = \{X_i\}_{i=1}^n$ representing string $T$, integer $q\geq 2$.}
  \KwOut{all $q$-grams and their frequencies which occur in $T$.}
  \SetKw{KwReport}{Report}
  Calculate $\VarOcc(X_i)$ for all $1\leq i\leq n$\; \label{algo:slpmain:varocc}  
  Calculate $\prefix(X_i,q-1)$ and $\suffix(X_i,q-1)$ for all $1\leq
  i\leq n-1$ \; \label{algo:slpmain:prefsuf}  
  $z \leftarrow \varepsilon$; $w \leftarrow []$\;
  \For{$i \leftarrow 1$ \KwTo $n$}{\label{algo:slpmain:mainloop}
    \If{$X_i = X_{\ell}X_r$ \KwAnd $|X_i| \geq q$}{
      $t_i = \suffix(X_\ell,q-1)\prefix(X_r,q-1)$;
      $z$.append($t_i$)\; \label{algo:slpmain:zappend}
      \lFor{$j \leftarrow 1$ \KwTo $|t_i|-q+1$}{
        $w$.append($\VarOcc(X_i)$)\;
      }
      \lFor{$j \leftarrow 1$ \KwTo $q-1$}{
        $w$.append(0)\; \label{algo:slpmain:append0}
      }
    }
  }
  \KwReport  
  $q$-gram frequencies in $z$, where each $q$-gram
  $z[i:i+q-1]$
  is {\em weighted} by $w[i]$.\label{algo:slpmainweightedfreqs}
\end{algorithm2e}

\begin{algorithm2e}[t]
  \caption{A variant of Algorithm~\ref{algo:naive} for weighted $q$-gram frequencies.}
  \label{algo:weighted_naive}
  \SetKwInput{KwOut}{Report}
  \SetKw{KwReport}{Report}
  \KwIn{string $T$, array of integers $w$ of length $|T|$, integer $q \geq 1$}
  \KwOut{$(P, \sum_{i\in\Occ(T,P)}w[i])$ for all $P\in\Sigma^q$ where
    $\sum_{i\in\Occ(T,P)}w[i] > 0$.}
  $\mathbf{S} \leftarrow \emptyset$\tcp*[l]{empty associative array}
  \For{$i\leftarrow 1$ \KwTo $|T|-q+1$}{
    $\mathit{qgram} \leftarrow T[i:i+q-1]$\;
    \lIf{$\mathit{qgram} \in\mathrm{keys}(\mathbf{S})$}{
      $\mathbf{S}[\mathit{qgram}] \leftarrow
      \mathbf{S}[\mathit{qgram}] + w[i]$\;      
    }
    \lElse{\lIf{$w[i]>0$}{$\mathbf{S}[\mathit{qgram}] \leftarrow w[i]$\tcp*[l]{new $q$-gram}}}
  }
  \lFor{$\mathit{qgram}\in\mathrm{keys}(\mathbf{S})$}{\KwReport{$(\mathit{qgram},\mathbf{S}[\mathit{qgram}])$}}
\end{algorithm2e}

\begin{algorithm2e}[t]
  \caption{A variant of Algorithm~\ref{algo:naive_sa} for weighted
    $q$-gram frequencies.}
  \label{algo:weighted_naive_sa}
  \SetKwInput{KwOut}{Output}
  \SetKw{KwReport}{Report}
  \SetKw{KwOr}{or}
  \SetKw{KwAnd}{and}
  \KwIn{string $T$, array of integers $w$ of length $|T|$, integer $q \geq 1$}
  \KwOut{$(i, \sum_{i\in\Occ(T,P)}w[i])$ for all $P\in\Sigma^q$ where
    $\sum_{i\in\Occ(T,P)}w[i] > 0$
    and some position $i\in\Occ(T,P)$.}
  $\SA \leftarrow \makeSuffixArray(T)$;
  $\LCP \leftarrow \makeLCPArray(T,SA)$;
  $\mathit{count} \leftarrow 1$\;  
  \For{$i\leftarrow 2$ \KwTo $|T|+1$}{
    \If{$i = |T|+1$ \KwOr $\LCP[i] < q$}{
      \lIf{$\mathit{count} > 0$}{
        \KwReport{$(\SA[i-1], \mathit{count})$}; $\mathit{count} \leftarrow 0$\;
      }

    }
    \lIf{$i \leq |T|$ \KwAnd $\SA[i] \leq |T| - q + 1$}{
      $\mathit{count} \leftarrow \mathit{count}+w[\SA[i]]$\;
    }
  }
\end{algorithm2e}

\section{Applications and Extensions}
We showed that for an SLP ${\mathcal{T}}$ of size $n$ representing string $T$,
$q$-gram frequency problems on $T$ can be reduced to
{\em weighted} $q$-gram frequency problems on a string $z$ of length $O(qn)$,
which can be much shorter than $T$.
This idea can further be applied to obtain efficient compressed
string processing algorithms for
interesting problems which we briefly introduce below.

\subsection{$q$-gram Spectrum Kernel}
A string kernel is a function that computes the inner product between two strings 
which are mapped to some feature space.
It is used when classifying string or text
data using methods such as Support Vector Machines (SVMs), and is usually the
dominating factor in the time complexity of SVM learning and classification.
A $q$-gram spectrum kernel~\cite{leslie02:_spect_kernel} considers the
feature space of $q$-grams.
For string $T$, let $\phi_q(T) = (|\Occ(T,p)|)_{p\in\Sigma^q}$.
The kernel function is defined as 
$K_q(T_1,T_2) = 
\langle\phi_q(T_1),\phi_q(T_2)\rangle =
\sum_{p\in\Sigma^q} |\Occ(T_1,p)| |\Occ(T_2,p)|$.
The calculation of the kernel function amounts to 
summing up the product of occurrence frequencies
in strings $T_1$ and $T_2$ for all 
$q$-grams which occur in both $T_1$ and $T_2$.
This can be done in $O(|T_1|+|T_2|)$ time using suffix arrays. 
For two SLPs ${\mathcal T_1}$ and ${\mathcal T_2}$ of 
size $n_1$ and $n_2$ representing strings $T_1$ and $T_2$,
respectively, the $q$-gram spectrum kernel $K_q(T_1,T_2)$ can be computed 
in $O(q(n_1+n_2))$ time by a slight modification of our algorithm.


\subsection{Optimal Substring Patterns of Length $q$}
Given two sets of strings, finding string patterns that are
frequent in one set and not in the other, is an important
problem in string data mining, with many problem formulations and the
types of patterns to be considered,
e.g.: in Bioinformatics~\cite{brazma98:_approac},
Machine Learning (optimal patterns~\cite{arimura98:_fast_algor_discov_optim_strin}),
and more recently KDD
(emerging patterns~\cite{chan03:_minin_emerg_subst}).
A simple optimal $q$-gram pattern discovery problem can be defined as follows:
Let  $\mathbf{T_1}$ 
and $\mathbf{T_2}$ 
be two multisets of strings.
The problem is to find the $q$-gram $p$ which gives
the highest (or lowest) score according to some scoring function
that depends only on $|\mathbf{T_1}|$, $|\mathbf{T_2}|$,
and the number of strings respectively in $\mathbf{T_1}$ and $\mathbf{T_2}$ for
which $p$ is a substring.
For uncompressed strings, the problem can be solved in $O(N)$ time,
where $N$ is the total length of the strings in both $\mathbf{T_1}$ and
$\mathbf{T_2}$, by applying the algorithm of~\cite{HuiCPM92} to two sets of strings.
For the SLP compressed version of this problem,
the input is two multisets of SLPs,
each representing strings in $\mathbf{T_1}$ and $\mathbf{T_2}$.
If $n$ is the total number of variables used in all
of the SLPs, the problem can be solved in $O(qn)$ time.


\subsection{Different Lengths}
The ideas in this paper can be used to consider all substrings of length {\em not only} $q$, 
but {\em all lengths up-to} $q$, with some modifications.
For the applications discussed above, 
although the number of such substrings increases to $O(q^2n)$,
the $O(qn)$ time complexity can be maintained
by using standard techniques of suffix
arrays~\cite{gusfield97:_algor_strin_trees_sequen,Kasai01}.
This is because there exist only $O(qn)$ substring with distinct frequencies
(corresponding to nodes of the suffix tree),
and the computations of the extra substrings can be summarized with
respect to them.

\ignore{
\subsection{Collage System}
Collage system~\cite{KidaCollageTCS} is a more general framework for modeling
various compression methods.
In addition to the simple concatenation operation used in SLPs,
it includes operations for repetition and prefix/suffix truncation 
of variables.
For example, while a LZ77 encoded representation of size $m$ may require
$O(m^2\log m)$ size when represented as an SLP, 
it can be represented as a collage system of size $O(m\log m)$~\cite{GasieniecSWAT96}.
Our algorithm can be extended to run in $O((q+h)n)$ time on
collage system of size $n$, where $h$ is the height of the derivation tree.
The increase in complexity comes from the handling of truncated
variables. Details will be presented in a forthcoming paper.
}

\section{Computational Experiments}
We implemented 4 algorithms (NMP, NSA, SMP, SSA) that 
count the frequencies of all $q$-grams in a given text.
NMP (Algorithm~\ref{algo:naive}) and NSA (Algorithm~\ref{algo:naive_sa})
work on the uncompressed text.
SMP (Algorithm~\ref{algo:slpmain} + 
Algorithm~\ref{algo:weighted_naive})
and
SSA (Algorithm~\ref{algo:slpmain} + 
Algorithm~\ref{algo:weighted_naive_sa})
work on SLPs.
The algorithms were implemented using the C++ language.
We used {\tt std::map} from the Standard Template Library (STL)
for the associative array implementation.
\footnote{We also used {\tt std::hash\_map} but omit the results due
  to lack of space. Choosing the hashing function to use is difficult,
  and we note that its performance was unstable and sometimes very bad
  when varying $q$.}
For constructing suffix arrays, we used the divsufsort
library\footnote{\url{http://code.google.com/p/libdivsufsort/}}
developed by Yuta Mori. This implementation is not linear time
in the worst case, but has been empirically shown to be one of
the fastest implementations on various data.

All computations were conducted on a Mac Xserve (Early 2009)
with 2 x 2.93GHz Quad Core Xeon processors and 24GB Memory,
only utilizing a single process/thread at once.
The program was compiled using the GNU C++ compiler ({\tt g++}) 4.2.1
with the {\tt -fast} option for optimization.
The running times are measured in seconds, starting from after reading
the uncompressed text into memory for NMP and NSA, and
after reading the SLP that represents the text into memory for SMP and
SSA.
Each computation is repeated at least 3 times, and the average is taken.
\subsection{Fibonacci Strings}

\begin{wrapfigure}[12]{r}{0.48\textwidth}
\vspace{-1.2cm}
  \begin{center}
    \includegraphics[width=0.47\textwidth]{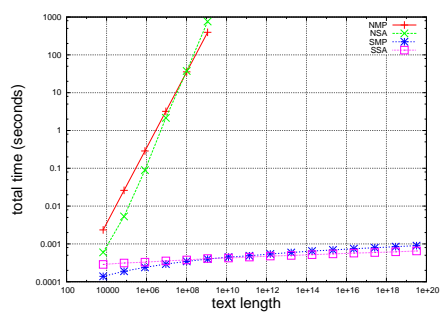}
    \caption{Running times of NMP, NSA, SMP, SSA on Fibonacci strings
      for $q=50$.}
    \label{fig:graph_fib}
  \end{center}
\end{wrapfigure}

The $i$ th Fibonacci string $F_i$ can be represented by the following
SLP:
$X_1 = \mathtt{b}$, $X_2 = \mathtt{a}$,
$X_i = X_{i-1} X_{i-2}$ for $i > 2$,
and $F_i = \derive(X_i)$.
Fig.~\ref{fig:graph_fib} shows the running times on
Fibonacci strings $F_{20}, F_{25}, \ldots, F_{95}$, for $q=50$.
Although this is an extreme case since Fibonacci strings can
be exponentially compressed, we can see that SMP and
SSA that work on the SLP are clearly faster than NMP and NSA which
work on the uncompressed string.

\subsection{Pizza \& Chili Corpus}
\begin{wrapfigure}[14]{r}{0.5\textwidth}
  \vspace{-1.0cm}
  \begin{center}
    \includegraphics[width=0.47\textwidth]{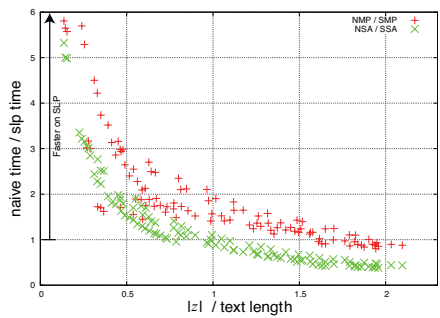}
    \caption{Time ratios NMP/SMP and NSA/SSA 
      plotted against ratio $|z|/|T|$.}
    \label{fig:speed_up}
  \end{center}
\end{wrapfigure}

We also applied the algorithms on texts
XML, DNA, ENGLISH, and PROTEINS, with sizes 50MB, 100MB, and 200MB,
obtained from the Pizza \& Chili
Corpus\footnote{\url{http://pizzachili.dcc.uchile.cl/texts.html}}.
We used RE-PAIR~\cite{LarssonDCC99} to obtain SLPs for this data.
%

Table~\ref{table:pizza} shows the running times for all algorithms and
data, where $q$ is varied from $2$ to $10$.
We see that for all corpora, 
SMP and SSA running on SLPs are actually faster than
NMP and NSA running on uncompressed text, when $q$ is small.
Furthermore, SMP is faster than SSA when $q$ is smaller.
Interestingly for XML, the SLP versions are faster even for $q$
up to $9$.

Fig.~\ref{fig:speed_up} shows the same results
as time ratio: NMP/SMP and NSA/ SSA, plotted against ratio:
(length of $z$ in Algorithm~\ref{algo:slpmain})/(length of uncompressed text).
As expected, the SLP versions are basically faster than their uncompressed
counterparts, when $|z|/\mbox{(text length)}$ is less than $1$,
since the SLP versions run the weighted versions of the uncompressed algorithms
on a text of length $|z|$.
SLPs generated by other grammar based compression algorithms showed
similar tendencies (data not shown).



\begin{table}[t]
  \caption{Running times in seconds for data from the Pizza \& Chili Corpus.
    Bold numbers represent the fastest time for each data and $q$.
    Times for SMP and SSA are prefixed with $\triangleright$,
    if they become fastest when all algorithms start from the SLP representation, 
    i.e., NMP and NSA require time for decompressing the SLP (denoted
    by decompression time).
    The bold horizontal lines show the boundary where $|z|$ in
    Algorithm~\ref{algo:slpmain}
    exceeds the uncompressed text length.
  }
  \label{table:pizza}
  \begin{center}
    \scriptsize
    \setlength{\tabcolsep}{1pt}
    \renewcommand{\rmdefault}{ptm}
    \renewcommand{\sfdefault}{phv}
    \renewcommand{\ttdefault}{pcr}
    \normalfont
    
\begin{tabular}{|c|r|r|r|r|r|r|r|r|r|r|r|r|r|r|r|}
 \hline
\multicolumn{16}{|c|}{XML} \\ \hline
& \multicolumn{5}{c|}{50MB} &  \multicolumn{5}{c|}{100MB} &  \multicolumn{5}{c|}{200MB} \\
& \multicolumn{3}{r}{SLP Size:} & \multicolumn{2}{r|}{2,702,383}& \multicolumn{3}{r}{SLP Size:} & \multicolumn{2}{r|}{5,059,578}& \multicolumn{3}{r}{SLP Size:} & \multicolumn{2}{r|}{9,541,590}\\& \multicolumn{3}{r}{decompression time:} & \multicolumn{2}{r|}{0.82 secs}& \multicolumn{3}{r}{decompression time:} & \multicolumn{2}{r|}{1.73 secs}& \multicolumn{3}{r}{decompression time:} & \multicolumn{2}{r|}{3.52 secs}\\ \hline
$q$ & \multicolumn{1}{c|}{$|z|$} & NMP & NSA & SMP & SSA & \multicolumn{1}{c|}{$|z|$} & NMP & NSA & SMP & SSA & \multicolumn{1}{c|}{$|z|$} & NMP & NSA & SMP & SSA\\ \hline
2 & 8,106,861 & 5.9 & 9.8 & \textbf{1.1} & 2.0 & 15,178,446 & 12.0 & 21.0 & \textbf{2.1} & 4.3 & 28,624,482 & 24.7 & 46.9 & \textbf{4.3} & 8.9 \\
\cline{1-1} \cline{2-6} \cline{7-11} \cline{12-16}
3 & 13,413,565 & 13.0 & 9.8 & \textbf{2.5} & 3.2 & 25,160,162 & 27.8 & 21.1 & \textbf{4.9} & 6.8 & 47,504,478 & 58.7 & 46.1 & \textbf{9.8} & 14.3 \\
\cline{1-1} \cline{2-6} \cline{7-11} \cline{12-16}
4 & 18,364,951 & 21.0 & 9.8 & 5.7 & \textbf{4.7} & 34,581,658 & 47.2 & 21.3 & 11.3 & \textbf{9.9} & 65,496,619 & 100.3 & 46.2 & 22.5 & \textbf{20.0} \\
\cline{1-1} \cline{2-6} \cline{7-11} \cline{12-16}
5 & 22,873,060 & 28.7 & 9.8 & 10.2 & \textbf{5.9} & 43,275,004 & 63.0 & 21.1 & 20.4 & \textbf{12.5} & 82,321,682 & 139.4 & 46.2 & 40.1 & \textbf{25.1} \\
\cline{1-1} \cline{2-6} \cline{7-11} \cline{12-16}
6 & 27,032,514 & 35.2 & 9.8 & 14.9 & \textbf{7.1} & 51,354,178 & 77.1 & 21.0 & 29.6 & \textbf{14.8} & 98,124,580 & 172.4 & 46.3 & 59.4 & \textbf{30.2} \\
\cline{1-1} \cline{2-6} \cline{7-11} \cline{12-16}
7 & 30,908,898 & 40.0 & 9.8 & 19.4 & \textbf{8.2} & 58,935,352 & 87.4 & 21.1 & 38.9 & \textbf{16.9} & 113,084,186 & 197.7 & 46.8 & 78.5 & \textbf{34.9} \\
\cline{1-1} \cline{2-6} \cline{7-11} \cline{12-16}
8 & 34,559,523 & 44.3 & 9.8 & 26.0 & \textbf{9.3} & 66,104,075 & 97.5 & 21.1 & 52.5 & \textbf{19.1} & 127,316,007 & 218.3 & 46.3 & 103.9 & \textbf{39.9} \\
\cline{1-1} \cline{2-6} \cline{7-11} \cline{12-16}
9 & 37,983,150 & 49.0 & \textbf{9.8} & 31.0 & $\triangleright$ 10.1  & 72,859,310 & 105.3 & 21.1 & 60.9 & \textbf{20.9} & 140,846,749 & 234.6 & 46.3 & 124.7 & \textbf{44.1} \\
\cline{1-1} \cline{2-6} \cline{7-11} \cline{12-16}
10 & 41,253,257 & 52.5 & \textbf{9.9} & 35.8 & 11.2 & 79,300,797 & 115.3 & \textbf{21.2} & 72.2 & $\triangleright$ 22.7  & 153,806,891 & 253.6 & \textbf{46.3} & 148.8 & $\triangleright$ 48.8  \\
\cline{1-1} \cline{2-6} \cline{7-11} \cline{12-16}
 \hline \hline
\multicolumn{16}{|c|}{DNA} \\ \hline
& \multicolumn{5}{c|}{50MB} &  \multicolumn{5}{c|}{100MB} &  \multicolumn{5}{c|}{200MB} \\
& \multicolumn{3}{r}{SLP Size:} & \multicolumn{2}{r|}{6,406,324}& \multicolumn{3}{r}{SLP Size:} & \multicolumn{2}{r|}{12,233,978}& \multicolumn{3}{r}{SLP Size:} & \multicolumn{2}{r|}{23,171,463}\\& \multicolumn{3}{r}{decompression time:} & \multicolumn{2}{r|}{1.23 secs}& \multicolumn{3}{r}{decompression time:} & \multicolumn{2}{r|}{2.54 secs}& \multicolumn{3}{r}{decompression time:} & \multicolumn{2}{r|}{5.21 secs}\\ \hline
$q$ & \multicolumn{1}{c|}{$|z|$} & NMP & NSA & SMP & SSA & \multicolumn{1}{c|}{$|z|$} & NMP & NSA & SMP & SSA & \multicolumn{1}{c|}{$|z|$} & NMP & NSA & SMP & SSA\\ \hline
2 & 19,218,924 & 2.2 & 13.7 & \textbf{1.9} & 5.7 & 36,701,886 & 4.7 & 30.5 & \textbf{3.9} & 12.6 & 69,514,341 & 9.8 & 70.0 & \textbf{8.0} & 26.1 \\
\cline{1-1} \cline{2-6} \cline{7-11} \cline{12-16}
3 & 32,030,826 & 4.4 & 13.7 & \textbf{3.0} & 8.6 & 61,169,030 & 9.1 & 30.5 & \textbf{5.8} & 18.6 & 115,856,038 & 18.7 & 70.1 & \textbf{11.8} & 38.8 \\
\cline{1-1} \cline{2-6} \cline{7-11} \cline{12-16}
4 & 44,833,624 & 6.5 & 13.7 & \textbf{4.5} & 12.3 & 85,624,856 & 13.4 & 30.5 & \textbf{8.9} & 25.3 & 162,182,697 & 28.0 & 70.0 & \textbf{17.6} & 52.9 \\
\cline{1-1} \wcline{2-6} \wcline{7-11} \cline{12-16}
5 & 57,554,843 & 8.6 & 13.8 & \textbf{6.7} & 15.5 & 109,976,706 & 17.8 & 30.5 & \textbf{13.1} & 32.3 & 208,371,656 & 37.0 & 69.9 & \textbf{26.3} & 67.9 \\
\cline{1-1} \cline{2-6} \cline{7-11} \wcline{12-16}
6 & 69,972,618 & 11.1 & 13.7 & \textbf{10.1} & 19.0 & 133,890,719 & 23.3 & 31.0 & \textbf{19.8} & 40.0 & 253,939,731 & 47.6 & 70.2 & \textbf{39.5} & 86.6 \\
\cline{1-1} \cline{2-6} \cline{7-11} \cline{12-16}
7 & 81,771,222 & 15.3 & \textbf{13.6} & $\triangleright$ 14.7  & 23.0 & 156,832,841 & 31.0 & 30.5 & \textbf{28.6} & 49.3 & 298,014,802 & 63.2 & 69.9 & \textbf{56.1} & 104.5 \\
\cline{1-1} \cline{2-6} \cline{7-11} \cline{12-16}
8 & 92,457,893 & 21.1 & \textbf{13.6} & 22.9 & 27.3 & 177,888,984 & 42.2 & \textbf{30.5} & 44.9 & 58.5 & 338,976,517 & 85.4 & \textbf{69.9} & 88.5 & 126.3 \\
\cline{1-1} \cline{2-6} \cline{7-11} \cline{12-16}
9 & 101,852,490 & 33.0 & \textbf{13.7} & 42.8 & 31.4 & 196,656,282 & 65.7 & \textbf{30.4} & 81.5 & 67.5 & 375,928,060 & 132.1 & \textbf{69.9} & 159.3 & 147.9 \\
\cline{1-1} \cline{2-6} \cline{7-11} \cline{12-16}
10 & 109,902,230 & 56.5 & \textbf{13.7} & 65.9 & 34.9 & 213,075,531 & 113.2 & \textbf{30.5} & 129.2 & 75.9 & 408,728,193 & 226.0 & \textbf{69.9} & 248.4 & 166.3 \\
\cline{1-1} \cline{2-6} \cline{7-11} \cline{12-16}
 \hline \hline
\multicolumn{16}{|c|}{ENGLISH} \\ \hline
& \multicolumn{5}{c|}{50MB} &  \multicolumn{5}{c|}{100MB} &  \multicolumn{5}{c|}{200MB} \\
& \multicolumn{3}{r}{SLP Size:} & \multicolumn{2}{r|}{4,861,619}& \multicolumn{3}{r}{SLP Size:} & \multicolumn{2}{r|}{10,063,953}& \multicolumn{3}{r}{SLP Size:} & \multicolumn{2}{r|}{18,945,126}\\& \multicolumn{3}{r}{decompression time:} & \multicolumn{2}{r|}{1.15 secs}& \multicolumn{3}{r}{decompression time:} & \multicolumn{2}{r|}{2.43 secs}& \multicolumn{3}{r}{decompression time:} & \multicolumn{2}{r|}{5.07 secs}\\ \hline
$q$ & \multicolumn{1}{c|}{$|z|$} & NMP & NSA & SMP & SSA & \multicolumn{1}{c|}{$|z|$} & NMP & NSA & SMP & SSA & \multicolumn{1}{c|}{$|z|$} & NMP & NSA & SMP & SSA\\ \hline
2 & 14,584,329 & 5.7 & 13.1 & \textbf{1.9} & 4.5 & 30,191,214 & 11.5 & 28.2 & \textbf{4.2} & 10.3 & 56,834,703 & 23.5 & 64.2 & \textbf{8.5} & 21.7 \\
\cline{1-1} \cline{2-6} \cline{7-11} \cline{12-16}
3 & 24,230,676 & 11.4 & 13.0 & \textbf{4.0} & 7.4 & 50,196,054 & 23.8 & 28.2 & \textbf{8.3} & 16.8 & 94,552,062 & 50.3 & 65.5 & \textbf{16.5} & 34.9 \\
\cline{1-1} \cline{2-6} \cline{7-11} \cline{12-16}
4 & 33,655,433 & 20.0 & 12.9 & \textbf{8.2} & 9.9 & 69,835,185 & 42.1 & 28.2 & \textbf{17.6} & 22.1 & 131,758,513 & 89.7 & 64.2 & \textbf{34.1} & 45.8 \\
\cline{1-1} \cline{2-6} \cline{7-11} \cline{12-16}
5 & 42,640,982 & 33.1 & 12.9 & 16.1 & \textbf{12.7} & 88,711,756 & 72.6 & \textbf{28.2} & 35.1 & $\triangleright$ 28.6  & 167,814,701 & 156.9 & 64.2 & 68.2 & \textbf{59.7} \\
\cline{1-1} \cline{2-6} \wcline{7-11} \cline{12-16}
6 & 51,061,064 & 49.5 & \textbf{12.9} & 27.1 & 15.5 & 106,583,131 & 111.8 & \textbf{28.5} & 59.7 & 35.3 & 202,293,814 & 240.8 & \textbf{64.4} & 116.1 & 74.3 \\
\cline{1-1} \wcline{2-6} \cline{7-11} \wcline{12-16}
7 & 58,791,311 & 65.1 & \textbf{12.9} & 40.1 & 18.4 & 123,180,654 & 143.6 & \textbf{28.3} & 88.3 & 42.3 & 234,664,404 & 313.7 & \textbf{64.3} & 173.5 & 90.3 \\
\cline{1-1} \cline{2-6} \cline{7-11} \cline{12-16}
8 & 65,777,414 & 79.6 & \textbf{12.9} & 59.1 & 20.8 & 138,382,443 & 176.8 & \textbf{28.3} & 131.3 & 48.5 & 264,668,656 & 385.9 & \textbf{64.8} & 256.7 & 104.5 \\
\cline{1-1} \cline{2-6} \cline{7-11} \cline{12-16}
9 & 71,930,623 & 92.7 & \textbf{12.9} & 74.2 & 23.0 & 152,010,306 & 207.8 & \textbf{28.5} & 166.0 & 54.2 & 291,964,684 & 454.6 & \textbf{64.5} & 335.0 & 118.0 \\
\cline{1-1} \cline{2-6} \cline{7-11} \cline{12-16}
10 & 77,261,995 & 105.3 & \textbf{13.0} & 89.7 & 25.1 & 164,021,382 & 235.9 & \textbf{28.4} & 205.2 & 59.8 & 316,387,791 & 521.2 & \textbf{64.7} & 425.3 & 131.4 \\
\cline{1-1} \cline{2-6} \cline{7-11} \cline{12-16}
 \hline \hline
\multicolumn{16}{|c|}{PROTEINS} \\ \hline
& \multicolumn{5}{c|}{50MB} &  \multicolumn{5}{c|}{100MB} &  \multicolumn{5}{c|}{200MB} \\
& \multicolumn{3}{r}{SLP Size:} & \multicolumn{2}{r|}{10,357,053}& \multicolumn{3}{r}{SLP Size:} & \multicolumn{2}{r|}{18,806,316}& \multicolumn{3}{r}{SLP Size:} & \multicolumn{2}{r|}{32,375,988}\\& \multicolumn{3}{r}{decompression time:} & \multicolumn{2}{r|}{1.67 secs}& \multicolumn{3}{r}{decompression time:} & \multicolumn{2}{r|}{3.51 secs}& \multicolumn{3}{r}{decompression time:} & \multicolumn{2}{r|}{7.05 secs}\\ \hline
$q$ & \multicolumn{1}{c|}{$|z|$} & NMP & NSA & SMP & SSA & \multicolumn{1}{c|}{$|z|$} & NMP & NSA & SMP & SSA & \multicolumn{1}{c|}{$|z|$} & NMP & NSA & SMP & SSA\\ \hline
2 & 31,071,084 & 4.5 & 14.5 & \textbf{4.0} & 10.2 & 56,418,873 & 9.0 & 32.2 & \textbf{7.6} & 20.4 & 97,127,889 & 18.0 & 69.0 & \textbf{13.6} & 38.0 \\
\cline{1-1} \cline{2-6} \cline{7-11} \cline{12-16}
3 & 51,749,628 & 9.4 & 14.5 & \textbf{7.6} & 16.2 & 93,995,974 & 18.7 & 32.1 & \textbf{14.1} & 32.3 & 161,825,337 & 37.3 & 69.0 & \textbf{25.5} & 60.0 \\
\cline{1-1} \wcline{2-6} \wcline{7-11} \wcline{12-16}
4 & 70,939,655 & 22.4 & \textbf{14.3} & 21.3 & 24.6 & 129,372,571 & 45.4 & \textbf{32.2} & 39.1 & 49.0 & 223,413,554 & 91.5 & \textbf{69.0} & $\triangleright$ 69.1  & 91.8 \\
\cline{1-1} \cline{2-6} \cline{7-11} \cline{12-16}
5 & 86,522,157 & 66.6 & \textbf{14.4} & 54.9 & 32.2 & 159,110,124 & 137.5 & \textbf{32.2} & 100.5 & 65.6 & 275,952,088 & 270.9 & \textbf{69.4} & 175.5 & 125.1 \\
\cline{1-1} \cline{2-6} \cline{7-11} \cline{12-16}
6 & 95,684,819 & 116.7 & \textbf{14.5} & 107.7 & 37.6 & 178,252,162 & 251.5 & \textbf{32.3} & 204.4 & 79.1 & 311,732,866 & 502.8 & \textbf{69.4} & 356.0 & 151.7 \\
\cline{1-1} \cline{2-6} \cline{7-11} \cline{12-16}
7 & 99,727,910 & 142.8 & \textbf{14.5} & 143.7 & 40.8 & 187,623,783 & 327.6 & \textbf{32.4} & 299.8 & 85.6 & 330,860,933 & 675.2 & \textbf{69.7} & 586.4 & 168.0 \\
\cline{1-1} \cline{2-6} \cline{7-11} \cline{12-16}
8 & 100,877,101 & 147.8 & \textbf{14.4} & 166.3 & 42.5 & 190,898,844 & 343.0 & \textbf{32.4} & 363.6 & 88.7 & 337,898,827 & 731.0 & \textbf{69.6} & 771.8 & 175.5 \\
\cline{1-1} \cline{2-6} \cline{7-11} \cline{12-16}
9 & 101,631,544 & 149.3 & \textbf{14.4} & 171.6 & 42.8 & 192,736,305 & 348.1 & \textbf{32.4} & 393.0 & 91.2 & 341,831,651 & 742.2 & \textbf{69.7} & 820.3 & 181.8 \\
\cline{1-1} \cline{2-6} \cline{7-11} \cline{12-16}
10 & 102,636,144 & 150.5 & \textbf{14.4} & 178.6 & 43.4 & 195,044,390 & 350.4 & \textbf{32.5} & 404.2 & 93.1 & 346,403,103 & 747.7 & \textbf{69.7} & 831.9 & 185.8 \\
\cline{1-1} \cline{2-6} \cline{7-11} \cline{12-16}

\end{tabular}

  \end{center}
\end{table}



\section{Conclusion}

We presented an $O(qn)$ time and space algorithm for calculating all
$q$-gram frequencies in a string, given an SLP of size $n$
representing the string.
This solves, much more efficiently, a more general problem than
considered in previous work.
Computational experiments on various real texts showed that the algorithms
run faster than algorithms that work on the
uncompressed string, when $q$ is small.
Although larger values of $q$ allow us to capture longer character dependencies,
the dimensionality of the features increases, making the space of
occurring $q$-grams sparse.
Therefore, meaningful values of $q$ for 
typical applications can be fairly small in practice (e.g. $3\sim 6$),
so our algorithms have practical value.

A future work is extending our algorithms that work on SLPs,
to algorithms that work on collage systems~\cite{KidaCollageTCS}.
A Collage System is a more general framework for modeling various
compression methods. 
In addition to the simple concatenation operation used in SLPs,
it includes operations for repetition and prefix/suffix truncation 
of variables.
\ignore{
For example, while a LZ77 encoded representation of size $m$ may require
$O(m^2\log m)$ size when represented as an SLP, 
it can be represented as a collage system of size $O(m\log m)$~\cite{GasieniecSWAT96}.
}

This is the first paper to show the potential of the compressed
string processing approach in developing efficient and {\em practical} algorithms
for problems in the field of string mining and classification.
More and more efficient algorithms for various processing of
text in compressed representations are becoming available.
We believe texts will eventually be stored in compressed form by default,
since not only will it save space, but it will also have the added benefit of being able to conduct
various computations on it more efficiently later on, when needed.



\bibliographystyle{splncs03}
\bibliography{ref}

\begin{thebibliography}{10}
\providecommand{\url}[1]{\texttt{#1}}
\providecommand{\urlprefix}{URL }

\bibitem{amir92:_effic_two_dimen_compr_match}
Amir, A., Benson, G.: Efficient two-dimensional compressed matching. In: Proc.
  Data Compression Conference 1992 (DCC '92). pp. 279--288 (1992)

\bibitem{arimura98:_fast_algor_discov_optim_strin}
Arimura, H., Wataki, A., Fujino, R., Arikawa, S.: A fast algorithm for
  discovering optimal string patterns in large text databases. In: Proc. 9th
  International Conference on Algorithmic Learning Theory (ALT '98). pp.
  247--261 (1998)

\bibitem{brazma98:_approac}
Brazma, A., Jonassen, I., Eidhammer, I., Gilbert, D.: Approaches to the
  automatic discovery of patterns in biosequences. J. Comp. Biol.  5(2),
  279--305 (1998)

\bibitem{chan03:_minin_emerg_subst}
Chan, S., Kao, B., Yip, C.L., Tang, M.: Mining emerging substrings. In: Proc.
  8th International Conference on Database Systems for Advanced Applications
  (DASFAA '03). p. 119 (2003)

\bibitem{charikar05:_small_gramm_probl}
Charikar, M., Lehman, E., Liu, D., Panigrahy, R., Prabhakaran, M., Sahai, A.,
  abhi shelat: The smallest grammar problem. IEEE Transactions on Information
  Theory  51(7),  2554--2576 (2005)

\bibitem{claudear:_self_index_gramm_based_compr}
Claude, F., Navarro, G.: Self-indexed grammar-based compression. Fundamenta
  Informaticae  (to appear), preliminary version: Proc. MFCS 2009 pp. 235--246

\bibitem{gusfield97:_algor_strin_trees_sequen}
Gusfield, D.: Algorithms on Strings, Trees, and Sequences. Cambridge University
  Press (1997)

\bibitem{hermelin09:_unified_algor_accel_edit_distan}
Hermelin, D., Landau, G.M., Landau, S., Weimann, O.: A unified algorithm for
  accelerating edit-distance computation via text-compression. In: Proc. STACS
  2009. pp. 529--540 (2009)

\bibitem{HuiCPM92}
Hui, L.C.K.: Color set size problem with application to string matching. In:
  Proc. CPM 1992. LNCS, vol. 644, pp. 230--243 (1992)

\bibitem{inenaga09:_findin_charac_subst_compr_texts}
Inenaga, S., Bannai, H.: Finding characteristic substring from compressed
  texts. In: Proc. The Prague Stringology Conference 2009. pp. 40--54 (2009)

\bibitem{Karkkainen_Sanders_icalp03}
K\"{a}rkk\"{a}inen, J., Sanders, P.: Simple linear work suffix array
  construction. In: Proc. ICALP 2003. LNCS, vol. 2719, pp. 943--955 (2003)

\bibitem{NJC97}
Karpinski, M., Rytter, W., Shinohara, A.: An efficient pattern-matching
  algorithm for strings with short descriptions. Nordic Journal of Computing
  4,  172--186 (1997)

\bibitem{Kasai01}
Kasai, T., Lee, G., Arimura, H., Arikawa, S., Park, K.: {Linear-time
  Longest-Common-Prefix Computation in Suffix Arrays and Its Applications}. In:
  Proc. CPM 2001. LNCS, vol. 2089, pp. 181--192 (2001)

\bibitem{KidaCollageTCS}
Kida, T., Shibata, Y., Takeda, M., Shinohara, A., Arikawa, S.: Collage system:
  A unifying framework for compressed pattern matching. Theoret. Comput. Sci.
  298,  253--272 (2003)

\bibitem{LarssonDCC99}
Larsson, N.J., Moffat, A.: Offline dictionary-based compression. In: Proc. Data
  Compression Conference 1999 (DCC '99). pp. 296--305 (1999)

\bibitem{leslie02:_spect_kernel}
Leslie, C., Eskin, E., Noble, W.S.: The spectrum kernel: A string kernel for
  {SVM} protein classification. In: Pacific Symposium on Biocomputing. vol.~7,
  pp. 566--575 (2002)

\bibitem{lifshits07:_proces_compr_texts}
Lifshits, Y.: Processing compressed texts: A tractability border. In: Proc. CPM
  2007. LNCS, vol. 4580, pp. 228--240 (2007)

\bibitem{manber93:_suffix}
Manber, U., Myers, G.: Suffix arrays: A new method for on-line string searches.
  SIAM J.~Computing  22(5),  935--948 (1993)

\bibitem{matsubara_tcs2009}
Matsubara, W., Inenaga, S., Ishino, A., Shinohara, A., Nakamura, T., Hashimoto,
  K.: Efficient algorithms to compute compressed longest common substrings and
  compressed palindromes. Theoret. Comput. Sci.  410(8--10),  900--913 (2009)

\bibitem{navarro07:_compr}
Navarro, G., M\"akinen, V.: Compressed full-text indexes. ACM Computing Surveys
   39(1), ~2 (2007)

\bibitem{SEQUITUR}
Nevill-Manning, C.G., Witten, I.H., Maulsby, D.L.: Compression by induction of
  hierarchical grammars. In: Proc. Data Compression Conference 1994 (DCC '94).
  pp. 244--253 (1994)

\bibitem{rytter03:_applic_lempel_ziv}
Rytter, W.: Application of {L}empel-{Z}iv factorization to the approximation of
  grammar-based compression. Theoret. Comput. Sci.  302(1--3),  211--222 (2003)

\bibitem{ShibataCIAC2000}
Shibata, Y., Kida, T., Fukamachi, S., Takeda, M., Shinohara, A., Shinohara, T.,
  Arikawa, S.: Speeding up pattern matching by text compression. In: Proc. 4th
  Italian Conference on Algorithms and Complexity (CIAC 2000). LNCS, vol. 1767,
  pp. 306--315 (2000)

\bibitem{LZ77}
Ziv, J., Lempel, A.: A universal algorithm for sequential data compression.
  IEEE Transactions on Information Theory  IT-23(3),  337--349 (1977)

\bibitem{LZ78}
Ziv, J., Lempel, A.: Compression of individual sequences via variable-length
  coding. IEEE Transactions on Information Theory  24(5),  530--536 (1978)

\end{thebibliography}
\end{document}